\documentclass[12pt]{article}

\usepackage[margin=1in]{geometry}
\usepackage{amsthm}
\usepackage{amssymb}
\usepackage{amsmath}
\usepackage{graphicx}
\usepackage{url}
\usepackage{xcolor}
\usepackage{framed}
\usepackage{colortbl}

\definecolor{Gray}{gray}{0.85}
\newcolumntype{a}{>{\columncolor{Gray}}r}

\newtheorem{theorem}{Theorem}
\newtheorem{corollary}[theorem]{Corollary}
\newtheorem{lemma}[theorem]{Lemma}
\newtheorem{proposition}[theorem]{Proposition}

\theoremstyle{definition}

\newtheorem{definition}[theorem]{Definition}

\title{A note on tight projective $2$-designs}
\author{
Joseph~W.~Iverson\footnote{Department of Mathematics, Iowa State University, Ames, IA} 
\and 
Emily~J.~King\footnote{Department of Mathematics, Colorado State University, Fort Collins, CO} 
\and 
Dustin~G.~Mixon\footnote{Department of Mathematics, The Ohio State University, Columbus, OH} \footnote{Translational Data Analytics Institute, The Ohio State University, Columbus, OH}
}
\date{}

\begin{document}
\maketitle

\begin{abstract}
We study tight projective $2$-designs in three different settings.
In the complex setting, Zauner's conjecture predicts the existence of a tight projective $2$-design in every dimension.
Pandey, Paulsen, Prakash, and Rahaman recently proposed an approach to make quantitative progress on this conjecture in terms of the entanglement breaking rank of a certain quantum channel.
We show that this quantity is equal to the size of the smallest weighted projective $2$-design.
Next, in the finite field setting, we introduce a notion of projective $2$-designs, we characterize when such projective $2$-designs are tight, and we provide a construction of such objects.
Finally, in the quaternionic setting, we show that every tight projective $2$-design for $\mathbb{H}^d$ determines an equi-isoclinic tight fusion frame of $d(2d-1)$ subspaces of $\mathbb{R}^{d(2d+1)}$ of dimension~$3$.
\end{abstract}

\section{Introduction}

Let $S(\mathbb{C}^d)$ denote the sphere of $x\in\mathbb{C}^d$ with $\|x\|_2^2=1$, and consider its uniform probability measure $\sigma$.
We let $\operatorname{Hom}_d(t)$ denote the complex vector space spanned by all monomial functions $\mathbb{C}^d\to\mathbb{C}$ that map $z=(z_1,\ldots,z_d)$ to $z_1^{\alpha_1}\cdots z_d^{\alpha_d}\overline{z}_1^{\beta_1}\cdots \overline{z}_d^{\beta_d}$ with $t=\sum_j\alpha_j=\sum_j\beta_j$.
Next, we take $\Pi_d^{(t)}$ to denote orthogonal projection onto the symmetric subspace $(\mathbb{C}^d)_{\operatorname{sym}}^{\otimes t}$ of $(\mathbb{C}^d)^{\otimes t}$.
Finally, put $[n]:=\{1,\ldots,n\}$.
Having established the neccesssary notation, a \textbf{projective $t$-design} for $\mathbb{C}^d$ is defined to be any $\{x_k\}_{k\in[n]}$ in $S(\mathbb{C}^d)$ that satisfies the following equivalent properties:

\begin{proposition}[see~\cite{Waldron:17}]
\label{prop.proj2design}
Given $\{x_k\}_{k\in[n]}$ in $S(\mathbb{C}^d)$ and $t\in\mathbb{N}$, the following are equivalent:
\begin{itemize}
\item[(a)]
$\frac{1}{n}\sum_{k\in[n]}p(x_k)=\int_{S(\mathbb{C}^d)}p(x)d\sigma(x)$ for every $p\in\operatorname{Hom}_d(t)$.
\item[(b)]
$\frac{1}{n}\sum_{k\in[n]}(x_k^{\otimes t})(x_k^{\otimes t})^*=\binom{d+t-1}{t}^{-1}\cdot \Pi_d^{(t)}$.
\item[(c)]
$\frac{1}{n^2}\sum_{k\in[n]}\sum_{\ell\in[n]}|\langle x_k,x_\ell\rangle|^{2t}=\binom{d+t-1}{t}^{-1}$.
\end{itemize}
\end{proposition}

In words, the cubature rule Proposition~\ref{prop.proj2design}(a) says that, for the purposes of integration over the sphere $S(\mathbb{C}^d)$, a projective $2$-design ``fools'' every $p\in\operatorname{Hom}_d(t)$ by mimicking the entire sphere.
Note that one may ``trace out'' one of the $t$ subsystems in Proposition~\ref{prop.proj2design}(b) to show that every projective $t$-design is also a projective $(t-1)$-design.
Proposition~\ref{prop.proj2design}(b) says that $\{x_k^{\otimes t}\}_{k\in[n]}$ forms what frame theorists would call a \textit{unit norm tight frame}~\cite{CasazzaK:12} for the $\binom{d+t-1}{t}$-dimensional complex Hilbert space $(\mathbb{C}^d)^{\otimes t}_{\operatorname{sym}}$.
With this perspective, Proposition~\ref{prop.proj2design}(c) corresponds to the \textit{frame potential}~\cite{BenedettoF:03} of $\{x_k^{\otimes t}\}_{k\in[n]}$.
One may generalize Proposition~\ref{prop.proj2design}(c) to obtain notions of projective $t$-designs for real, quaternionic, and octonionic spaces~\cite{Munemasa:07}.
The complex case with $t=2$ is particularly relevant in quantum state tomography~\cite{Scott:06}, where it is desirable to take $n$ as small as possible.
This motivates the following result, which refers to $\{x_k\}_{k\in[n]}$ in $S(\mathbb{C}^d)$ as \textbf{equiangular} if
\[
|\{|\langle x_k,x_\ell\rangle|^2:k,\ell\in[n],k\neq \ell\}|
=1.
\]

\begin{proposition}[special case of Proposition~1.1 in~\cite{BannaiH:85}]
\label{prop.d^2 threshold}
Consider $X=\{x_k\}_{k\in[n]}$ in $S(\mathbb{C}^d)$.
\begin{itemize}
\item[(a)]
If $X$ is a projective $2$-design, then $n\geq d^2$ with equality precisely when $X$ is equiangular.
\item[(b)]
If $X$ is equiangular, then $n\leq d^2$ with equality precisely when $X$ is a projective $2$-design.
\end{itemize}
\end{proposition}

In the case of equality $n=d^2$, $\{x_k\}_{k\in[n]}$ is known as a \textbf{tight projective $2$-design} for $\mathbb{C}^d$, which corresponds to an object in quantum physics known as a \textit{symmetric, informationally complete positive operator--valued measure}~\cite{Caves:99}.
In his Ph.D.\ thesis~\cite{Zauner:99}, Zauner conjectured that for every $d>1$, there exists a tight projective $2$-design for $\mathbb{C}^d$ of a particular form.
To date, there are only finitely many $d\in\mathbb{N}$ for which a tight projective $2$-design is known to exist~\cite{FuchsHS:17,Grassl:online,Flammia:online}, and the conjecture is apparently related to the Stark conjectures in algebraic number theory~\cite{Kopp:19}.
A solution to Zauner's conjecture will be rewarded with a $2021$~EUR prize from the National Quantum Information Centre in Poland~\cite{KCIK:online}.

Pandey, Paulsen, Prakash, and Rahaman~\cite{PandeyPPR:20} recently proposed a new approach to Zauner's conjecture.
They identify an explicit quantum channel $\mathfrak{Z}_d\colon\mathbb{C}^{d\times d}\to\mathbb{C}^{d\times d}$ whose so-called \textit{entanglement breaking rank} is $\operatorname{ebr}(\mathfrak{Z}_d)\geq d^2$, and furthermore, equality holds if and only if there exists a tight projective $2$-design for $\mathbb{C}^d$.
As such, any new upper bound on this entanglement breaking rank represents quantitative progress towards Zauner's conjecture.
In addition, Pandey et al.\ consider various analytic approaches to obtain such bounds in small dimenions.

As a consequence of Proposition~\ref{prop.d^2 threshold}, every tight projective $2$-design for $\mathbb{C}^d$ is necessarily an \textbf{equiangular tight frame (ETF)}, that is, a unit norm tight frame that is also equiangular.
ETFs correspond to optimal codes in projective space that achieve equality in the \textit{Welch bound}~\cite{Welch:74} (also known as the \textit{simplex bound}~\cite{ConwayHS:96}).
By virtue of this optimality, ETFs find applications in wireless communication~\cite{StrohmerH:03}, compressed sensing~\cite{BandeiraFMW:13}, and digital fingerprinting~\cite{MixonQKF:13}.
Motivated by these applications, many ETFs were recently constructed using various mixtures of algebra and combinatorics~\cite{StrohmerH:03,XiaZG:05,DingF:07,FickusMT:12,FickusJMP:18,BodmannK:20,IversonJM:20b,IversonJM:20,IversonM:online,IversonM:online2}; see~\cite{FickusM:15} for a survey.
Despite this flurry of work, several problems involving ETFs (such as Zauner's conjecture) remain open, and a finite field model was recently proposed to help study these remaining problems~\cite{GreavesIJM:20,GreavesIJM:20b}.

Notice that if $\{x_k\}_{k\in[n]}$ is a tight projective $2$-design for $\mathbb{C}^d$, then Propositions~\ref{prop.proj2design} and~\ref{prop.d^2 threshold} together imply that $\{x_k^{\otimes2}\}_{k\in[n]}$ is an ETF for $(\mathbb{C}^d)^{\otimes2}_{\operatorname{sym}}$.
This suggests another approach to Zauner's conjecture~\cite{ApplebyBFG:19} in which one seeks ETFs of $d^2$ vectors in the $\binom{d+1}{2}$-dimensional complex Hilbert space $(\mathbb{C}^d)^{\otimes2}_{\operatorname{sym}}$.
There are several known constructions of such ETFs~\cite{FickusM:15,BodmannK:20,IversonJM:20b,IversonJM:20}, but in order to correspond to a tight projective $2$-design, the ETF must consist of rank-$1$ symmetric tensors.

We note that one may leverage linear programming bounds to obtain analogous results to Proposition~\ref{prop.d^2 threshold} that relate projective $t$-designs over different spaces to different sized angle sets~\cite{BannaiH:85}.
In the real case, tight projective $2$-designs have size $\binom{d+1}{2}$ and are only known to exist for $d\in\{2,3,7,23\}$, with $d=119$ being the smallest dimension for which existence is currently unknown; see~\cite{LemmensS:73,Makhnev:02,BannaiMV:05,NebeV:13,Gillespie:18}.
In the quaternion case, tight projective $2$-designs have size $2d(d-1)$; they are only known to exist for $d\in\{2,3\}$~\cite{CohnKM:16,EtTaoui:20}, and there is numerical evidence that they do not exist for $d\in\{4,5\}$~\cite{CohnKM:16}.
The octonions are only capable of supporting a projective space for $d\in\{2,3\}$, and tight projective $2$-designs exist in both cases~\cite{CohnKM:16}.

In this paper, we study tight projective $2$-designs in three different settings.
In Section~2, we consider the complex setting, specifically, the new quantitative approach of Pandey et al.~\cite{PandeyPPR:20}.
Here, we show that $\operatorname{ebr}(\mathfrak{Z}_d)$ is precisely the size of the smallest \textit{weighted} projective $2$-design for $\mathbb{C}^d$.
This identification allows us to find new upper bounds on $\operatorname{ebr}(\mathfrak{Z}_d)$.
Next, in Section~3, we use Proposition~\ref{prop.proj2design}(b) to find an analog of projective $2$-designs in a finite field setting.
This continues the line of inquiry from~\cite{GreavesIJM:20,GreavesIJM:20b} of tackling hard problems from frame theory in a finite field model.
In this setting, we obtain an analog of Proposition~\ref{prop.d^2 threshold}, and then we construct a family of tight projective $2$-designs.
Finally in Section~4, we consider the quaternionic setting, where we take inspiration from the fact that a tight projective $2$-design for $\mathbb{C}^d$ can be used to produce an ETF of $d^2$ vectors in $(\mathbb{C}^d)^{\otimes2}_{\operatorname{sym}}$.
In particular, we show how a tight projective $2$-design for $\mathbb{H}^d$ can be used to produce an equi-isoclinic tight fusion frame of $2d(d-1)$ different $3$-dimensional subspaces of the $d(2d+1)$-dimensional real Hilbert space of $d\times d$ quaternionic anti-Hermitian matrices.

\section{The complex setting}

A linear map $\Phi\colon\mathbb{C}^{d\times d}\to \mathbb{C}^{m\times m}$ is said to be \textbf{entanglement breaking} if it admits an \textbf{entanglement breaking decomposition}:
\begin{equation}
\label{eq.ebd}
\Phi(X)
=\sum_{k\in[n]} R_kXR_k^*,
\qquad
\sum_{k\in[n]} R_k^*R_k=I_d,
\qquad
\operatorname{rank}R_k=1
~~
\text{for every}
~~
k\in[n].
\end{equation}
The \textbf{entanglement breaking rank} of $\Phi$, denoted by $\operatorname{ebr}(\Phi)$, is the smallest $n$ for which there exists $\{R_k\}_{k\in[n]}$ in $\mathbb{C}^{m\times d}$ satisfying \eqref{eq.ebd}.
Let $\{e_i\}_{i\in[d]}$ denote the standard basis in $\mathbb{C}^d$.
The \textbf{Choi matrix} of $\Phi$ is given by
\[
C_\Phi
:=\sum_{i\in[d]}\sum_{j\in[d]}e_ie_j^*\otimes\Phi(e_ie_j^*)
\in \mathbb{C}^{d\times d} \otimes \mathbb{C}^{m\times m}.
\]
In words, $C_\Phi$ is a $d\times d$ block array whose $(i,j)$th block is $\Phi(e_ie_j^*)\in\mathbb{C}^{m\times m}$.
One may use $C_\Phi$ to discern useful properties about $\Phi$.
For example, $\Phi$ is completely positive when $C_\Phi$ is positive semidefinite; see Theorem~2.22 in~\cite{Watrous:18}.

We are interested in the \textit{quantum depolarizing channel} $\mathfrak{Z}_d\colon\mathbb{C}^{d\times d}\to\mathbb{C}^{d\times d}$ defined by
\[
\mathfrak{Z}_d(X)
:=\frac{1}{d+1}\Big(X+\operatorname{tr}X\cdot I_d\Big).
\]
One may verify that $\mathfrak{Z}_d$ is entanglement breaking as a consequence of its scaled Choi matrix $\frac{1}{d} C_{\mathfrak{Z}_d}$ being a separable bipartite state (specifically, the \textit{isotropic state} with $\lambda=1/d$ from Example~7.25 in~\cite{Watrous:18}).
Pandey, Paulsen, Prakash, and Rahaman~\cite{PandeyPPR:20} pointed to this quantum channel as an opportunity for quantitative progress on Zauner's conjecture:

\begin{proposition}[cf.\ Corollary~III.3 and Theorem~V.3 in~\cite{PandeyPPR:20}]\
\label{prop.ebrZ}
\begin{itemize}
\item[(a)]
$\operatorname{ebr}(\mathfrak{Z}_d)\geq d^2$, with equality if and only if there exists a tight projective $2$-design for $\mathbb{C}^d$.
\item[(b)]
$\operatorname{ebr}(\mathfrak{Z}_d)\leq d^2+d$ whenever $d$ is a prime power.
\end{itemize}
\end{proposition}

We say unit vectors $\{x_k\}_{k\in[n]}$ in $\mathbb{C}^d$ form a \textbf{weighted projective $t$-design} if there exist weights $\{w_k\}_{k\in[n]}$ such that
\[
\sum_{k\in[n]}w_k(x_k^{\otimes t})(x_k^{\otimes t})^*
=\tbinom{d+t-1}{t}^{-1}\cdot\Pi_d^{(t)},
\qquad
\sum_{k\in[n]}w_k=1,
\qquad
w_k\geq0,
\qquad
k\in[n].
\]
Notice that every projective $t$-design is a weighted projective $t$-design with weights $w_k=1/n$.
In addition, it is known that every weighted projective $2$-design for $\mathbb{C}^d$ has size at least $d^2$, and if equality holds, then the weights are all $1/n$; see Theorem~4 in~\cite{Scott:06}.
What follows is the main result of this section, of which Proposition~\ref{prop.ebrZ} is a corollary:

\begin{theorem}
\label{thm.weighted2designtoebr}
The smallest weighted projective $2$-design for $\mathbb{C}^{d}$ has size $\operatorname{ebr}(\mathfrak{Z}_d)$.
\end{theorem}

In fact, one may use Theorem~\ref{thm.weighted2designtoebr} to improve upon Proposition~\ref{prop.ebrZ}(b) by collecting various weighted $2$-designs from the literature.
Specifically, Theorem~4.1, Proposition~4.2, and Corollary~4.2 in~\cite{RoyS:07}, and Corollaries~4.4 and~4.6 in~\cite{BodmannH:16} give the following: 

\begin{corollary}\
\begin{itemize}
\item[(a)]
$\operatorname{ebr}(\mathfrak{Z}_d)\leq kd^2+2d$ whenever $kd+1$ is a prime power with $k\in\mathbb{N}$.
\item[(b)]
$\operatorname{ebr}(\mathfrak{Z}_d)\leq d^2+(p+1)d$ whenever $d+1=p^k$ with $p$ prime and $k\in\mathbb{N}$.
\item[(c)]
$\operatorname{ebr}(\mathfrak{Z}_d)\leq d^2+1$ whenever $d-1$ is a prime power.
\item[(d)]
$\operatorname{ebr}(\mathfrak{Z}_d)\leq d^2+d-1$ whenever $d$ is a prime power.
\end{itemize}
\end{corollary}

Since $\mathfrak{Z}_d$ is entanglement breaking, Theorem~\ref{thm.weighted2designtoebr} also implies the existence of weighted projective $2$-designs; we note that this also follows from the main result in~\cite{SeymourZ:84}.

\begin{corollary}
For each $d\in\mathbb{N}$, there is a weighted projective $2$-design for $\mathbb{C}^d$ of size $\binom{d+1}{2}^2$.
\end{corollary}

\begin{proof}
Since $\mathfrak{Z}_d$ is entanglement breaking, Theorem~\ref{thm.weighted2designtoebr} promises a weighted projective $2$-design $\{x_k\}_{k\in[n]}$ for $\mathbb{C}^d$.
Then $\Pi_d^{(2)}$ resides in the conic hull of $\{(x_k^{\otimes2})(x_k^{\otimes 2})^*\}_{k\in[n]}$, which in turn is contained in the $\binom{d+1}{2}^2$-dimensional real vector space of Hermitian operators over $(\mathbb{C}^{d})^{\otimes2}_{\operatorname{sym}}$.
By Carath\'{e}odory's theorem, there exists $S\subseteq[n]$ with $|S|\leq \binom{d+1}{2}^2$ and weights $w_k\geq0$ for $k\in S$ such that
\[
\sum_{k\in S}w_k(x_k^{\otimes2})(x_k^{\otimes2})^*
=\tbinom{d+1}{2}^{-1}\cdot\Pi_d^{(2)}.
\]
Furthermore, taking the trace of both sides reveals that $\sum_{k\in S}w_k=1$.
As such, $\{x_k\}_{k\in S}$ is a weighted projective $2$-design.
\end{proof}

The remainder of this section proves Theorem~\ref{thm.weighted2designtoebr}.
We first collect a few helpful lemmas.
Let $T\colon\mathbb{C}^{d\times d}\to\mathbb{C}^{d\times d}$ denote the tranposition operator defined by $T(X):=X^\top$.

\begin{lemma}[cf.\ Proposition~III.5 in~\cite{PandeyPPR:20}]
\label{lem.ebrofTofZ}
It holds that $\operatorname{ebr}(\Phi)=\operatorname{ebr}(T\circ\Phi)$ with
\[
\Phi(X)
=\sum_{k\in[n]} x_ky_k^\top X(x_ky_k^\top)^*
\qquad
\Longleftrightarrow
\qquad
(T\circ\Phi)(X)
=\sum_{k\in[n]} \overline{x}_ky_k^\top X(\overline{x}_ky_k^\top)^*.
\]
\end{lemma}

\begin{proof}
The claim follows from the following manipulation:
\[
\bigg(\sum_{k\in[n]} x_ky_k^\top X\overline{y}_k\overline{x}_k^\top\bigg)^\top
=\sum_{k\in[n]} \overline{x}_k(\overline{y}_k^\top X^\top y_k)x_k^\top
=\sum_{k\in[n]} \overline{x}_k(\overline{y}_k^\top X^\top y_k)^\top x_k^\top
=\sum_{k\in[n]} \overline{x}_ky_k^\top X \overline{y}_kx_k^\top,
\]
where the second step takes the transpose of a scalar.
Indeed, we have both $(x_ky_k^\top)^*=\overline{y}_k\overline{x}_k^\top$ and $(\overline{x}_ky_k^\top)^*=\overline{y}_kx_k^\top$.
\end{proof}

Both of the following lemmas were implicitly used in the proof of Corollary~III.7 in~\cite{PandeyPPR:20}.
To prove them, we will repeatedly use the following identity, which is valid for any linear $\Phi\colon\mathbb{C}^{d\times d}\to\mathbb{C}^{m\times m}$, and any $w,y\in\mathbb{C}^d$ and $x,z\in\mathbb{C}^m$:
\begin{align}
(w\otimes x)^\top C_\Phi (y\otimes z)
\nonumber
&=\sum_{i\in[d]}\sum_{j\in[d]} w_iy_j\cdot x^\top \Phi(e_ie_j^*) z\\
\label{eq.choiid}
&=x^\top \Phi\bigg(\sum_{i\in[d]}\sum_{j\in[d]} w_iy_j\cdot e_ie_j^*\bigg) z
=x^\top \Phi(wy^\top) z.
\end{align}

\begin{lemma}
\label{lem.decompvschoi}
An entanglement breaking map $\Phi$ has entanglement breaking decomposition
\[
\Phi(X)
=\sum_{k\in[n]}a_kb_k^\top X (a_kb_k^\top)^*.
\]
if and only if $\Phi$ has Choi matrix
\[
C_\Phi
=\sum_{k\in[n]}(b_k\otimes a_k)(b_k\otimes a_k)^*.
\]
\end{lemma}

\begin{proof}
($\Rightarrow$)
For any $w,x,y,z$, we may apply \eqref{eq.choiid} to get
\begin{align*}
(w\otimes x)^\top C_\Phi (y\otimes z)
&=x^\top \bigg(\sum_{k\in[n]}a_kb_k^\top (wy^\top) \overline{b}_k\overline{a}_k^\top\bigg) z\\
&=\sum_{k\in[n]}(x^\top a_kb_k^\top w)(y^\top \overline{b}_k\overline{a}_k^\top z)\\
&=\sum_{k\in[n]}(w\otimes x)^\top(b_k\otimes a_k)(\overline{b}_k\otimes\overline{a}_k)^\top(y\otimes z)\\
&=(w\otimes x)^\top\bigg(\sum_{k\in[n]}(b_k\otimes a_k)(b_k\otimes a_k)^*\bigg)(y\otimes z).
\end{align*}
Since $w,x,y,z$ are arbitrary, the result follows.

($\Leftarrow$)
For any $w,x,y,z$, we may similarly apply \eqref{eq.choiid} to get
\begin{align*}
x^\top\Phi(wy^\top)z
&=(w\otimes x)^\top C_\Phi (y\otimes z)\\
&=(w\otimes x)^\top\bigg(\sum_{k\in[n]}(b_k\otimes a_k)(b_k\otimes a_k)^*\bigg)(y\otimes z)\\
&=x^\top \bigg(\sum_{k\in[n]}a_kb_k^\top (wy^\top) \overline{b}_k\overline{a}_k^\top\bigg) z.
\end{align*}
Since $w,x,y,z$ are arbitrary, the result follows.
\end{proof}

\begin{lemma}
\label{lem.choiofTofZ}
$C_{T\circ \mathfrak{Z}_d}=\frac{2}{d+1}\Pi_d^{(2)}$.
\end{lemma}

\begin{proof}
For any $w,x,y,z$, we may apply \eqref{eq.choiid} to get
\begin{align*}
(w\otimes x)^\top C_{T\circ \mathfrak{Z}_d}(y\otimes z)
&=x^\top \Big[(T\circ\mathfrak{Z}_d)(wy^\top)\Big] z\\
&=x^\top\frac{1}{d+1}\Big(yw^\top +\operatorname{tr}wy^\top\cdot I_d\Big)z\\
&=\frac{1}{d+1}\Big(w^\top z\cdot x^\top y + w^\top y \cdot x^\top z\Big)
=(w\otimes x)^\top\frac{1}{d+1}\Big(z\otimes y+y\otimes z\Big).
\end{align*}
Since $w,x$ are arbitrary, it follows that
\[
C_{T\circ \mathfrak{Z}_d}(y\otimes z)
=\frac{2}{d+1}\cdot\frac{1}{2}(z\otimes y+y\otimes z)
=\frac{2}{d+1}\Pi_d^{(2)}(y\otimes z).
\]
Since $y,z$ are arbitrary, the result follows.
\end{proof}

We are now ready to prove the main result of this section.

\begin{proof}[Proof of Theorem~\ref{thm.weighted2designtoebr}]
Suppose $\{x_k\}_{k\in[n]}$ is a weighted projective $2$-design for $\mathbb{C}^{d}$ with weights $\{w_k\}_{k\in[n]}$.
We claim that $\operatorname{ebr}(\mathfrak{Z}_d)\leq n$.
To see this, first recall that
\[
\sum_{k\in[n]}w_k(x_k\otimes x_k)(x_k\otimes x_k)^*=\frac{2}{d(d+1)}\cdot\Pi_d^{(2)}.
\]
Taking $a_k:=x_k$ and $b_k:=\sqrt{dw_k}x_k$ then gives
\[
\sum_{k\in[n]}(b_k\otimes a_k)(b_k\otimes a_k)^*
=d\sum_{k\in[n]}w_k(x_k\otimes x_k)(x_k\otimes x_k)^*
=\frac{2}{d+1}\cdot\Pi_d^{(2)}
=C_{T\circ\mathfrak{Z}_d},
\]
where the last step applies Lemma~\ref{lem.choiofTofZ}.
Lemmas~\ref{lem.ebrofTofZ} and~\ref{lem.decompvschoi} then imply that
\[
\operatorname{ebr}(\mathfrak{Z}_d)
=\operatorname{ebr}(T\circ\mathfrak{Z}_d)
\leq n.
\]
Next, suppose $n=\operatorname{ebr}(\mathfrak{Z}_d)=\operatorname{ebr}(T\circ\mathfrak{Z}_d)$ and consider entanglement breaking decomposition
\[
(T\circ\Phi)(X)
=\sum_{k\in[n]}a_kb_k^\top X (a_kb_k^\top)^*.
\]
Then Lemmas~\ref{lem.decompvschoi} and~\ref{lem.choiofTofZ} together imply
\[
\sum_{k\in[n]}(b_k\otimes a_k)(b_k\otimes a_k)^*
=\frac{2}{d+1}\cdot\Pi_d^{(2)}.
\]
Notice that for every antisymmetric $x\in(\mathbb{C}^d)^{\otimes 2}$, it holds that
\[
\sum_{k\in[n]}|\langle b_k\otimes a_k,x\rangle|^2
=x^*\bigg(\sum_{k\in[n]}(b_k\otimes a_k)(b_k\otimes a_k)^*\bigg)x
=x^*\frac{2}{d+1}\Pi_d^{(2)}x
=0.
\]
It follows that each $b_k\otimes a_k$ is necessarily symmetric, and therefore takes the form $\sqrt{dw_k}x_k\otimes x_k$ for some unit vector $x_k$ and scalar $w_k\geq0$.
Then
\[
\sum_{k\in[n]}w_k(x_k\otimes x_k)(x_k\otimes x_k)^*=\frac{2}{d(d+1)}\cdot\Pi_d^{(2)}.
\]
The fact that $\sum_{k\in[n]}w_k=1$ follows from taking the trace of both sides.
Overall, there exists a weighted projective $2$-design for $\mathbb{C}^{d}$ of size $n=\operatorname{ebr}(\mathfrak{Z}_d)$, as claimed.
\end{proof}

\section{The finite field setting}

In this section, we introduce a notion of projective $2$-designs in a finite field setting.
Here, we will find an analog to Proposition~\ref{prop.d^2 threshold} in which tight projective $2$-designs are identified as maximal systems of equiangular lines, and then we will provide several examples.
We start by reviewing some preliminaries; the reader is encouraged to see~\cite{GreavesIJM:20} for more information.

Let $q$ be a prime power.
Given $a \in \mathbb{F}_{q^2}$, we abbreviate $\overline{a} = a^q$ for its image under the Frobenius automorphism fixing $\mathbb{F}_q \leq \mathbb{F}_{q^2}$.
The conjugate transpose of a matrix $A$ is denoted by $A^*$.
We consider $\mathbb{F}_{q^2}^d$ under the nondegenerate Hermitian form $\langle x , y \rangle = x^* y$, which is notably conjugate-linear in the first variable.
A subspace $V \leq \mathbb{F}_{q^2}^d$ is called \textbf{nondegenerate} if $V \cap V^\perp = \{ 0 \}$, where 
\[ V^\perp := \{ x \in \mathbb{F}_{q^2}^d : \langle x, y \rangle = 0 \text{ for every }y \in V \}. \]
In that case, every $x \in \mathbb{F}_{q^2}^d$ can be written uniquely as $x = Px + Qx$ with $Px \in V$ and $Qx \in V^\perp$, where $P \colon \mathbb{F}_{q^2}^d \to \mathbb{F}_{q^2}^d$ is \textbf{orthogonal projection} onto $V$.

\begin{definition}
Let $V \leq \mathbb{F}_{q^2}^d$ be nondegenerate.
We say $\{ x_k \}_{k \in [n]}$ in $V$ is a $c$-\textbf{tight frame} for $V$ with constant $c\in \mathbb{F}_q$ if
\begin{itemize}
\item[(i)]
$\operatorname{span} \{x_k\}_{k \in [n]} = V$, and
\item[(ii)]
$\sum_{k \in [n]} \langle x_k, y \rangle x_k = cy$ for every $y \in V$.
\end{itemize}
For $a \in \mathbb{F}_q$, a $c$-tight frame is an \textbf{equal-norm tight frame}, or $(a,c)$-NTF, if
\begin{itemize}
\item[(iii)]
$\langle x_k , x_k \rangle = a$ for every $k \in [n]$.
\end{itemize}
For $b \in \mathbb{F}_q$, an $(a,c)$-NTF is an \textbf{equiangular tight frame}, or $(a,b,c)$-ETF, if
\begin{itemize}
\item[(iv)]
$\langle x_k, x_\ell \rangle \langle x_\ell, x_k \rangle = b$ for every $k,\ell\in[n]$ with $k \neq \ell$.
\end{itemize}
Meanwhile, an $(a,b)$-\textbf{equiangular system} in $V$ satisfies (iii) and (iv), but not necessarily (i) or (ii).
\end{definition}

Notice that (ii) implies (i) if $c \neq 0$.
Furthermore, if $P$ is orthogonal projection onto $V$, then (ii) is equivalent to
\begin{itemize}
\item[(ii$'$)]
$\sum_{k \in [n]} x_k x_k^* = cP$.
\end{itemize}

We will repeatedly make use of the following basic results from~\cite{GreavesIJM:20}:

\begin{proposition}[Corollary~3.8 in~\cite{GreavesIJM:20}]
\label{prop.vanishing frame bound redundancy}
If $V \leq \mathbb{F}_{q^2}^d$ is nondegenerate and $\{x_k \}_{k\in [n]}$ is a tight frame for $V$ with constant $c = 0$, then $n \geq 2\dim V$.
\end{proposition}

\begin{proposition}[Equation~(3.2) and Proposition~4.7 in~\cite{GreavesIJM:20}]\
\label{prop.parameter relations}
\begin{itemize}
\item[(a)]
If $\{x_k \}_{k\in [n]}$ is an $(a,c)$-NTF for $V$, then $na = c \dim V$.
\item[(b)]
If $\{x_k \}_{k\in [n]}$ is an $(a,b,c)$-ETF for $V$, then $a(c-a) = (n-1) b$.
\end{itemize}
\end{proposition}

\begin{proposition}[Gerzon's bound, see Theorem~4.1 in~\cite{GreavesIJM:20} and its proof]
\label{prop.gerzon}
If $\{x_k \}_{k \in [n]}$ is an $(a,b)$-equiangular system in $\mathbb{F}_{q^2}^d$ and $a^2 \neq b$, then $n \leq d^2$.
If equality holds and $a \neq 0$, then $\{ x_kx_k^* \}_{k\in [n]}$ is a basis for the $\mathbb{F}_q$-linear space
\[
\{ X \in \mathbb{F}_{q^2}^{d\times d} : X = X^* \}.
\]
If equality holds and $a = 0$, then the $\mathbb{F}_q$-span of $\{ x_k x_k^* \}_{ k \in [n]}$ is the subspace
\[ 
\{ X \in \mathbb{F}_{q^2}^{d\times d} : X = X^*,~\operatorname{tr} X = 0 \}, 
\]
and $\sum_{k \in [n]} x_k x_k^* = 0$ is the unique $\mathbb{F}_q$-linear dependency of $\{ x_k x_k^* \}_{k \in [n]}$ up to a scalar.
\end{proposition}

\subsection{Projective 2-designs}

Throughout this subsection, we assume $q$ is odd.
Let $e_1,\dotsc,e_d \in \mathbb{F}_{q^2}^d$ denote the standard basis.
Then $\{ e_i \otimes e_j \}_{i,j \in [d]}$ is a basis for $(\mathbb{F}_{q^2}^d)^{\otimes 2}$.
We write 
\[ 
(\mathbb{F}_{q^2}^d)^{\otimes 2}_{\operatorname{sym}} 
:= \bigg\{ \sum_{i \in [d]} \sum_{j \in [d]} c_{ij} (e_i \otimes e_j) : c_{ij} = c_{ji} \text{ for every } i,j \in [d] \bigg\} 
\]
for the subspace of symmetric tensors, and we define
\[
\Pi_d^{(2)} 
:= \frac{1}{2} \sum_{i \in [d]} \sum_{j\in [d]} \Big(e_i e_i^* \otimes e_j e_j^* + e_i e_j^* \otimes e_j e_i^*\Big) 
\in (\mathbb{F}_{q^2}^{d \times d})^{\otimes 2}.
\]

\begin{lemma}
\label{lem.symmetric tensors are nondegenerate}
$(\mathbb{F}_{q^2}^d)^{\otimes 2}_{\operatorname{sym}} \leq (\mathbb{F}_{q^2}^{d})^{\otimes2}$ is nondegenerate, and $\Pi_d^{(2)}$ is its orthogonal projection.
\end{lemma}

\begin{proof}
For nondegeneracy, let $y = \sum_{i \in [d]}\sum_{j \in [d]} a_{ij} (e_i \otimes e_j) \in (\mathbb{F}_{q^2}^d)^{\otimes 2}_{\operatorname{sym}}$ be nonzero.
Then there exist $k,\ell \in [d]$ such that $a_{k\ell} = a_{\ell k} \neq 0$.
It follows that
\begin{align*}
\langle y, e_k \otimes e_\ell + e_\ell \otimes e_k \rangle &=
\sum_{i \in [d]} \sum_{j \in [d]} \overline{a_{ij}} \langle e_i \otimes e_j, e_k \otimes e_\ell \rangle
+
\sum_{i \in [d]} \sum_{j \in [d]} \overline{a_{ij}} \langle e_i \otimes e_j, e_\ell \otimes e_k \rangle
= 2\overline{a_{k\ell}},
\end{align*}
which is nonzero by assumption.
Thus, $y$ is not orthogonal to $(\mathbb{F}_{q^2}^d)^{\otimes 2}_{\operatorname{sym}}$.

Next, we show that $\Pi_d^{(2)}$ projects orthogonally onto $(\mathbb{F}_{q^2}^d)^{\otimes 2}_{\operatorname{sym}}$.
To this end, choose any vector $x = \sum_{k \in [d]}\sum_{\ell \in [d]} c_{k\ell} (e_k \otimes e_\ell)$ and compute
\begin{align*}
\Pi_d^{(2)} x 
&= \frac{1}{2} \sum_{i \in [d]}\sum_{j \in [d]}\sum_{k \in [d]}\sum_{\ell \in [d]} c_{k\ell} \Big( (e_i e_i^* \otimes e_j e_j^*)(e_k \otimes e_\ell) + (e_i e_j^* \otimes e_j e_i^*)(e_k \otimes e_\ell) \Big) \\
&= \frac{1}{2} \sum_{i \in [d]}\sum_{j \in [d]}\sum_{k \in [d]}\sum_{\ell \in [d]} c_{k\ell} \Big( (e_i e_i^*)e_k \otimes (e_j e_j^*)e_\ell + (e_i e_j^*)e_k \otimes (e_j e_i^*)e_\ell \Big) \\
&= \frac{1}{2} \sum_{k \in [d]}\sum_{\ell \in [d]} c_{k\ell} \bigg(  \sum_{i \in [d]}\sum_{j \in [d]} (e_i e_i^*)e_k \otimes (e_j e_j^*)e_\ell  +  \sum_{i \in [d]}\sum_{j \in [d]} (e_i e_j^*)e_k \otimes (e_j e_i^*)e_\ell  \bigg) \\
&= \frac{1}{2} \sum_{k \in [d]}\sum_{\ell \in [d]} c_{k\ell} (e_k \otimes e_\ell + e_\ell \otimes e_k),
\end{align*}
which belongs to $(\mathbb{F}_{q^2}^d)^{\otimes 2}_{\operatorname{sym}}$.
Then
\[
x - \Pi_d^{(2)} x 
= \frac{1}{2} \sum_{k \in [d]}\sum_{\ell \in [d]} c_{k\ell}(e_k \otimes e_\ell - e_\ell \otimes e_k),
\]
and it is straightforward to check this is orthogonal to every $y \in (\mathbb{F}_{q^2}^d)^{\otimes 2}_{\operatorname{sym}}$.
As such, $x = (\Pi_d^{(2)} x) + (x - \Pi_d^{(2)} x)$ gives the desired decomposition of $x$.
\end{proof}

As a consequence of Lemma~\ref{lem.symmetric tensors are nondegenerate}, we may consider tight frames over $(\mathbb{F}_{q^2}^d)^{\otimes2}_{\operatorname{sym}}$.
This allows us to define the following analog of Proposition~\ref{prop.proj2design}(b):

\begin{definition}
$\{x_k \}_{k \in[n]}$ in $\mathbb{F}_{q^2}^d$ is an $(a,c_1,c_2)$-\textbf{projective $2$-design} if $a,c_1,c_2 \in \mathbb{F}_q$ and
\begin{itemize}
\item[(i)]
$\langle x_k, x_k \rangle = a$ for every $k \in [n]$,
\item[(ii)]
$\{x_k\}_{k\in[n]}$ is a $c_1$-tight frame for $\mathbb{F}_{q^2}^d$, and
\item[(iii)]
$\{ x_k^{\otimes 2} \}_{k\in [n]}$ is a $c_2$-tight frame for $(\mathbb{F}_{q^2}^d)^{\otimes 2}_{\operatorname{sym}}$.
\end{itemize}
\end{definition}

With this definition, the finite field setting enjoys an analogy to Proposition~\ref{prop.d^2 threshold}:

\begin{theorem}
\label{thm.finite field bound}
If $\{x_k\}_{k\in[n]}$ in $\mathbb{F}_{q^2}^d$ is a projective $2$-design, then $n\geq d^2$.
\end{theorem}

To prove this theorem, we will repeatedly make use of the following:

\begin{lemma}
\label{lem.decomposition of A}
If $\{x_k\}_{k\in[n]}$ in $\mathbb{F}_{q^2}^d$ is an $(a,c_1,c_2)$-projective $2$-design with $c_2\neq0$, then for every $A\in\mathbb{F}_{q^2}^{d\times d}$, it holds that
\[
A
=\frac{2}{c_2}\sum_{k\in[n]}x_kx_k^*Ax_kx_k^*-\operatorname{tr}A\cdot I_d.
\]
\end{lemma}

\begin{proof}
To prove the result, we multiply both sides of the identity
$
\sum_{k\in[n]}(x_k^{\otimes2})(x_k^{\otimes2})^*
=c_2\cdot \Pi_d^{(2)}
$
by $A\otimes I_d$ and then ``trace out'' the first subsystem.
Explicitly, we define the \textit{partial trace} $\operatorname{tr}_1\colon\mathbb{F}_{q^2}^{d\times d}\otimes\mathbb{F}_{q^2}^{d\times d}\to\mathbb{F}_{q^2}^{d\times d}$ by taking
\[
\operatorname{tr}_1(A\otimes B)
:=\operatorname{tr}(A)\cdot B
\]
and extending linearly.
Since $\{x_k\}_{k\in[n]}$ is a projective $2$-design, we have
\begin{equation}
\label{eq.to trace out}
\operatorname{tr}_1\bigg(\sum_{k\in[n]}(x_k^{\otimes2})(x_k^{\otimes2})^*(A\otimes I_d)\bigg)
=\operatorname{tr}_1\Big(c_2\cdot \Pi_d^{(2)}\cdot (A\otimes I_d)\Big).
\end{equation}
We cycle the trace to simplify the left-hand side of \eqref{eq.to trace out}:
\begin{align}
\operatorname{tr}_1\bigg(\sum_{k\in[n]}(x_k^{\otimes2})(x_k^{\otimes2})^*(A\otimes I_d)\bigg)
\nonumber
&=\operatorname{tr}_1\bigg(\sum_{k\in[n]}(x_kx_k^*)^{\otimes2}(A\otimes I_d)\bigg)\\
\nonumber
&=\operatorname{tr}_1\bigg(\sum_{k\in[n]}x_kx_k^*A\otimes x_kx_k^*\bigg)\\
\label{eq.LHS of trace out}
&=\sum_{k\in[n]}\operatorname{tr}(x_kx_k^*A)\cdot x_kx_k^*
=\sum_{k\in[n]}x_k(x_k^*Ax_k)x_k^*.
\end{align}
For the right-hand side of \eqref{eq.to trace out}, we apply the definition of $\Pi_d^{(2)}$:
\begin{align}
\operatorname{tr}_1\Big(c_2\cdot \Pi_d^{(2)}\cdot (A\otimes I_d)\Big)
\nonumber
&=\operatorname{tr}_1\bigg(\frac{c_2}{2}\sum_{i\in[d]}\sum_{j\in[d]}\Big(e_ie_i^*\otimes e_je_j^*+e_ie_j^*\otimes e_je_i^*\Big)(A\otimes I_d)\bigg)\\
\nonumber
&=\frac{c_2}{2}\sum_{i\in[d]}\sum_{j\in[d]}\Big(\operatorname{tr}(e_ie_i^*A)\cdot e_je_j^*+\operatorname{tr}(e_ie_j^*A)\cdot e_je_i^*\Big)\\
\label{eq.RHS of trace out}
&=\frac{c_2}{2}\Big(\operatorname{tr}A\cdot I_d+A\Big).
\end{align}
The result follows by equating \eqref{eq.LHS of trace out} to \eqref{eq.RHS of trace out} and rearranging.
\end{proof}

\begin{proof}[Proof of Theorem~\ref{thm.finite field bound}]
Suppose $\{x_k\}_{k\in[n]}$ in $\mathbb{F}_{q^2}^d$ is an $(a,c_1,c_2)$-projective $2$-design.

\textbf{Case I:}
$c_2=0$.
Then $\{x_kx_k^*\}_{k\in[n]}$ is an $(a^2,0)$-NTF for $(\mathbb{F}_{q^2}^d)_{\operatorname{sym}}^{\otimes 2}$.
By Proposition~\ref{prop.vanishing frame bound redundancy},
\[
n
\geq 2\cdot\operatorname{dim}((\mathbb{F}_{q^2}^d)_{\operatorname{sym}}^{\otimes 2})
=d^2+d
>d^2.
\]

\textbf{Case II:}
$c_1\neq 0$ and $c_2\neq 0$.
Apply Lemma~\ref{lem.decomposition of A} and the identity $\sum_{k\in[n]}x_kx_k^*=c_1\cdot I_d$:
\begin{align*}
A
&=\frac{2}{c_2}\sum_{k\in[n]}x_kx_k^*Ax_kx_k^*-\operatorname{tr}A\cdot I_d\\
&=\frac{2}{c_2}\sum_{k\in[n]}x_kx_k^*Ax_kx_k^*-\operatorname{tr}A\cdot \frac{1}{c_1}\sum_{k\in[n]}x_kx_k^*
=\sum_{k\in[n]}\bigg(\frac{2}{c_2}x_k^*Ax_k-\frac{1}{c_1}\operatorname{tr}A\bigg)x_kx_k^*
\end{align*}
for every $A\in\mathbb{F}_{q^2}^{d\times d}$.
It follows that the $\mathbb{F}_{q^2}$-span of $\{x_kx_k^*\}_{k\in[n]}$ is $\mathbb{F}_{q^2}^{d\times d}$, and so $n\geq d^2$.

\textbf{Case III:}
$c_1=0$ and $c_2\neq 0$.
Lemma~\ref{lem.decomposition of A} implies that the $\mathbb{F}_{q^2}$-span of $\{x_kx_k^*\}_{k\in[n]}\cup\{I_d\}$ is $\mathbb{F}_{q^2}^{d\times d}$, while the identity $\sum_{k\in[n]}x_kx_k^*=c_1\cdot I_d=0$ implies that $\{x_kx_k^*\}_{k\in[n]}$ is linearly dependent.
Since $\{x_kx_k^*\}_{k\in[n]}\cup\{I_d\}$ is a linearly dependent spanning set of size $n+1$, it follows that $n+1\geq d^2+1$, i.e., $n\geq d^2$.
\end{proof}

\begin{theorem}
\label{thm.characterization of tight designs}
Any two of the following statements together imply the third statement:
\begin{itemize}
\item[(a)]
$\{x_k\}_{k\in[n]}$ in $\mathbb{F}_{q^2}^d$ is a projective $2$-design.
\item[(b)]
$n=d^2$.
\item[(c)]
There exist $a,b,c_1\in\mathbb{F}_q$ such that
\begin{itemize}
\item[(i)]
$a^2\neq b$,
\item[(ii)]
$a^2-b=\frac{bc_1}{a}$ if $a\neq0$,
\item[(iii)]
$d\equiv-1\bmod p$ if $a=0$, and
\item[(iv)]
$\{x_k\}_{k\in[n]}$ in $\mathbb{F}_{q^2}^d$ is an $(a,b,c_1)$-equiangular tight frame.
\end{itemize}
\end{itemize}
When (a), (b), and (c) hold, $\{x_k\}_{k\in[n]}$ is an $(a,c_1,c_2)$-projective $2$-design with $c_2=2(a^2-b)$.
\end{theorem}

To prove Theorem~\ref{thm.characterization of tight designs}, we need a method of demonstrating that a collection of vectors forms a projective $2$-design.
For this, we will apply the following:

\begin{lemma}
\label{lem.ebr-esque move}
Take $\mathbb{F}$ to be $\mathbb{C}$ or $\mathbb{F}_{q^2}$.
Given $\{x_k\}_{k\in[n]}$ in $\mathbb{F}^d$, define $\Psi\colon\mathbb{F}^{d\times d}\to\mathbb{F}^{d\times d}$ by
\[
\Psi(A)
:=\sum_{k\in[n]}x_kx_k^*A^*x_kx_k^*.
\]
Then
\[
\sum_{k\in[n]}(x_k^{\otimes2})(x_k^{\otimes2})^*
=\sum_{i\in[d]}\sum_{j\in[d]}e_ie_j^*\otimes \Psi(e_ie_j^*).
\]
\end{lemma}

To prove Lemma~\ref{lem.ebr-esque move}, we will use lemmas from the previous section, with the appropriate interpretation of conjugation in the case $\mathbb{F}=\mathbb{F}_{q^2}$; indeed, the proofs of these results are valid under this interpretation.

\begin{proof}[Proof of Lemma~\ref{lem.ebr-esque move}]
Consider the linear map $\overline{\Psi}\colon\mathbb{F}^{d\times d}\to\mathbb{F}^{d\times d}$ defined by $\overline{\Psi}(A):=\overline{\Psi(A)}$.
Then
\[
(T\circ\overline{\Psi})(A)
=\Psi(A)^*
=\sum_{k\in[n]}x_kx_k^*Ax_kx_k^*
=\sum_{k\in[n]}x_k\overline{x}_k^\top A(x_k\overline{x}_k^\top)^*.
\]
Lemma~\ref{lem.ebrofTofZ} then gives
\[
\overline{\Psi}(A)
=\sum_{k\in[n]}\overline{x}_k\overline{x}_k^\top A(\overline{x}_k\overline{x}_k^\top)^*.
\]
Finally, we apply Lemma~\ref{lem.decompvschoi} to get
\[
\sum_{k\in[n]}(\overline{x}_k^{\otimes2})(\overline{x}_k^{\otimes2})^*
=\sum_{i\in[d]}\sum_{j\in[d]}e_ie_j^*\otimes \overline{\Psi}(e_ie_j^*),
\]
and we take conjugates of both sides to obtain the result.
\end{proof}

\begin{proof}[Proof of Theorem~\ref{thm.characterization of tight designs}]
First, (a)$\wedge$(c)$\Rightarrow$(b) follows from Theorem~\ref{thm.finite field bound} and Proposition~\ref{prop.gerzon}.

Next, we demonstrate (a)$\wedge$(b)$\Rightarrow$(c) by considering each case in the proof of Theorem~\ref{thm.finite field bound}.

\textbf{Case I:}
$c_2=0$.
This case does not occur since $n\geq d^2+d$ implies $n\neq d^2$.

For the remaining cases, we have $c_2\neq0$.
For these cases, we will use the fact that $a=0$ if and only if $c_1=0$.
To see this, apply Lemma~\ref{lem.decomposition of A} with $A=I_d$ and the identity $\sum_{k\in[n]}x_kx_k^*=c_1\cdot I_d$ to get
\[
I_d
=\frac{2}{c_2}\sum_{k\in[n]}x_kx_k^*I_dx_kx_k^*-\operatorname{tr}I_d\cdot I_d
=\frac{2a}{c_2}\sum_{k\in[n]}x_kx_k^*-d\cdot I_d
=\bigg(\frac{2ac_1}{c_2}-d\bigg)\cdot I_d
\]
If $0\in\{a,c_1\}$, then the above identity implies $d\equiv-1\bmod p$ and $n=d^2\equiv 1\bmod p$.
Furthermore, since $\{x_k\}_{k\in[n]}$ is an $(a,c_1)$-NTF, Proposition~\ref{prop.parameter relations}(a) gives that $na=dc_1$.
If $0\in\{a,c_1\}$, then squaring both sides gives $a^2=n^2a^2=d^2c_1^2=c_1^2$.
It follows that $a=0$ if and only if $c_1=0$, as claimed.
Furthermore, $d\equiv-1\bmod p$ if $a=0$.

\textbf{Case II:}
$c_1\neq0$ and $c_2\neq 0$.
For every $A\in\mathbb{F}_{q^2}^{d\times d}$, Lemma~\ref{lem.decomposition of A} and the identity $\sum_{k\in[n]}x_kx_k^*=c_1\cdot I_d$ together imply
\begin{equation}
\label{eq.unique decomposition}
A
=\sum_{k\in[n]}\bigg(\frac{2}{c_2}x_k^*Ax_k-\frac{1}{c_1}\operatorname{tr}A\bigg)x_kx_k^*.
\end{equation}
Since $\{x_kx_k^*\}_{k\in[n]}$ is a spanning set of $\mathbb{F}_{q^2}^{d\times d}$ of size $n=d^2=\operatorname{dim}(\mathbb{F}_{q^2}^{d\times d})$, it is also a basis.
Then the decomposition \eqref{eq.unique decomposition} is unique.
For $A:=x_\ell x_\ell^*$, this implies
\begin{equation}
\label{eq.equate unique coefficients}
\frac{2}{c_2}\langle x_k,x_\ell\rangle^{q+1}-\frac{a}{c_1}
=\frac{2}{c_2}x_k^*Ax_k-\frac{1}{c_1}\operatorname{tr}A
=\left\{\begin{array}{cl}
1&\text{if }k=\ell\\
0&\text{if }k\neq\ell.
\end{array}\right.
\end{equation}
It follows that
\[
\langle x_k,x_\ell\rangle^{q+1}
=\frac{ac_2}{2c_1}
=:b
\]
whenever $k\neq\ell$, i.e., $\{x_k\}_{k\in[n]}$ is an $(a,b,c_1)$-ETF.
Then \eqref{eq.equate unique coefficients} gives
\begin{align}
\label{eq.unique1}
\frac{2}{c_2}a^2-\frac{a}{c_1}
&=1,\\
\label{eq.unique2}
\frac{2}{c_2}b-\frac{a}{c_1}
&=0.
\end{align}
Subtract \eqref{eq.unique2} from \eqref{eq.unique1} and rearrange to get
\begin{equation}
\label{eq.unique3}
a^2-b
=\frac{c_2}{2}
\neq0.
\end{equation}
Finally, since $a\neq 0$, we may rearrange \eqref{eq.unique2} to get $\frac{c_2}{2}=\frac{bc_1}{a}$, which combined with \eqref{eq.unique3} gives $a^2-b=\frac{bc_1}{a}$, as claimed.

\textbf{Case III:}
$c_1=0$ and $c_2\neq0$.
We claim that the $\mathbb{F}_{q^2}$-span of $\{x_kx_k^*\}_{k\in[n]}$ equals the $(d^2-1)$-dimensional subspace $\{X\in\mathbb{F}_{q^2}^{d\times d}:\operatorname{tr}X=0\}$.
Indeed, the inclusion $\subseteq$ follows from the fact that  $\operatorname{tr}(x_kx_k^*)=a=0$ for every $k\in[n]$.
The reverse inclusion follows from a dimension count, since Lemma~\ref{lem.decomposition of A} implies that the $\mathbb{F}_{q^2}$-span of $\{x_kx_k^*\}_{k\in[n]}\cup\{I_d\}$ is all of $\mathbb{F}_{q^2}^{d\times d}$.
Next, since $n=(d^2-1)+1$, it follows that the identity $\sum_{k\in[n]}x_kx_k^*=c_1\cdot I_d$ gives the only linear dependency of $\{x_kx_k^*\}_{k\in[n]}$ up to scalar multiplication, namely, $\sum_{k\in[n]}x_kx_k^*=0$.
Taking $A:=x_\ell x_\ell^*$ in Lemma~\ref{lem.decomposition of A} gives
\[
x_\ell x_\ell^*
=\frac{2}{c_2}\sum_{k\in[n]}x_kx_k^*x_\ell x_\ell^*x_kx_k^*-\operatorname{tr}x_\ell x_\ell^*\cdot I_d
=\frac{2}{c_2}\sum_{k\in[n]}\langle x_k,x_\ell\rangle^{q+1} x_kx_k^*,
\]
and rearranging gives
\begin{equation}
\label{eq.dependency}
\sum_{k\in[n]}z_{k\ell}x_kx_k^*=0,
\end{equation}
where
\[
z_{k\ell}
:=\left\{\begin{array}{ll}
\frac{2}{c_2}\langle x_k,x_\ell\rangle^{q+1}&\text{if }k\neq\ell\\
\frac{2}{c_2}\langle x_\ell,x_\ell\rangle^{q+1}-1&\text{if }k=\ell\\
\end{array}\right\}
=\left\{\begin{array}{cl}
\frac{2}{c_2}\langle x_k,x_\ell\rangle^{q+1}&\text{if }k\neq\ell\\
-1&\text{if }k=\ell.\\
\end{array}\right.
\]
Since $\sum_{k\in[n]}x_kx_k^*=0$ is the unique dependency up to scaling, the dependency \eqref{eq.dependency} requires $z_{k\ell}=-1$ for every $k\neq\ell$, i.e.,
\[
\langle x_k,x_\ell\rangle^{q+1}
=-\frac{c_2}{2}
=:b.
\]
As such, $\{x_k\}_{k\in[n]}$ is an $(a,b,c_1)$-ETF.
Since $c_2\neq0$, we have $b\neq0=a^2$, as claimed.

Finally, we demonstrate (b)$\wedge$(c)$\Rightarrow$(a) with the help of Lemma~\ref{lem.ebr-esque move}.
To this end, we consider the linear map $\Psi^*\colon A\mapsto \Psi(A)^*$.
Since $\{x_k\}_{k\in[n]}$ is an $(a,b,c_1)$-ETF, then 
\begin{align}
\Psi^*(x_\ell x_\ell^*)
=\sum_{k\in[n]} x_kx_k^* x_\ell x_\ell^* x_kx_k^*
\nonumber
&=\sum_{k\in[n]}\langle x_k,x_\ell\rangle^{q+1} x_kx_k^*\\
\label{eq.psistar}
&=a^2 x_\ell x_\ell^* + \sum_{\substack{k\in[n]\\k\neq\ell}}bx_kx_k^*
=(a^2-b) x_\ell x_\ell^* + bc_1\cdot I_d.
\end{align}
This expression obfuscates the linearity of $\Psi^*$, which we elucidate in two separate cases.

\textbf{Case I:}
$a\neq0$.
Since $\operatorname{tr}(x_\ell x_\ell^*)=a$, we may continue \eqref{eq.psistar}:
\[
\Psi^*(x_\ell x_\ell^*)
=(a^2-b) x_\ell x_\ell^* + \frac{bc_1}{a}\cdot\operatorname{tr}(x_\ell x_\ell^*)\cdot I_d.
\]
Since $0\neq a^2\neq b$, it follows from Proposition~\ref{prop.gerzon} that the $\mathbb{F}_{q^2}$-span of $\{x_kx_k^*\}_{k\in[n]}$ has dimension $\operatorname{dim}_{\mathbb{F}_q}\{X\in\mathbb{F}_{q^2}^{d\times d}:X=X^*\}=d^2$, and so it equals $\mathbb{F}_{q^2}^{d\times d}$.
Thus, we may linearly extend the above identity to get
\[
\Psi^*(A)
=(a^2-b) A + \frac{bc_1}{a}\cdot\operatorname{tr}A \cdot I_d
\]
for every $A\in\mathbb{F}_{q^2}^{d\times d}$.
In particular, we have
\begin{align*}
\Psi(e_ie_j^*)
=\Big(\Psi^*(e_ie_j^*)\Big)^*
&=\Big((a^2-b) e_ie_j^* + \frac{bc_1}{a}\cdot\operatorname{tr}(e_ie_j^*) \cdot I_d\Big)^*\\
&=(a^2-b) e_je_i^* + \frac{bc_1}{a}\cdot\delta_{ij}\cdot I_d
=(a^2-b)\cdot\Big(e_je_i^*+\delta_{ij}\cdot I_d\Big),
\end{align*}
where the last step applies our assumption that $a^2-b=\frac{bc_1}{a}$.
Then Lemma~\ref{lem.ebr-esque move} gives
\begin{align*}
\sum_{k\in[n]}(x_k^{\otimes2})(x_k^{\otimes2})^*
&=\sum_{i\in[d]}\sum_{j\in[d]}e_ie_j^*\otimes \Psi(e_ie_j^*)\\
&=2(a^2-b)\sum_{i\in[d]}\sum_{j\in[d]} e_ie_j^*\otimes \frac{1}{2}\Big( e_je_i^*+\delta_{ij}\cdot I_d \Big)
=2(a^2-b)\cdot\Pi_d^{(2)}.
\end{align*}
Since $c_2:=2(a^2-b)\neq0$ by assumption, it follows that $\{x_k\}_{k\in[n]}$ is an $(a,c_1,c_2)$-projective $2$-design, as desired.

\textbf{Case II:}
$a=0$.
Then $d\equiv-1\bmod p$ by assumption.
Furthermore, since $\{x_k\}_{k\in[n]}$ is an $(a,c_1)$-NTF, Proposition~\ref{prop.gerzon}(a) gives that $0=na=dc_1=-c_1$, i.e., $c_1=0$.
With this, we continue \eqref{eq.psistar}:
\begin{equation}
\label{eq.psistar2}
\Psi^*(x_\ell x_\ell^*)
=(a^2-b) x_\ell x_\ell^* + bc_1\cdot I_d
=-b x_\ell x_\ell^*.
\end{equation}
By equality in Gerzon's bound, Proposition~\ref{prop.gerzon} implies that the $\mathbb{F}_{q^2}$-span of $\{x_kx_k^*\}_{k\in[n]}$ has dimension $\operatorname{dim}_{\mathbb{F}_q}\{X\in\mathbb{F}_{q^2}^{d\times d}:X=X^*,~\operatorname{tr}X=0\}=d^2-1$, and therefore equals $\{X\in\mathbb{F}_{q^2}^{d\times d}:\operatorname{tr}X=0\}$.
By linearity, \eqref{eq.psistar2} implies $\Psi^*(A)=-bA$ for every $A\in \mathbb{F}_{q^2}^{d\times d}$ with $\operatorname{tr}A=0$.
In order to determine $\Psi^*(A)$ for all $A$, we also consider
\[
\Psi^*(I_d)
=\sum_{k\in[n]}x_kx_k^* I_d x_kx_k^*
=a\sum_{k\in[n]}x_kx_k^*
=0.
\]
Since $\{x_kx_k^*\}_{k\in[n]}\cup\{I_d\}$ spans $\mathbb{F}_{q^2}^{d\times d}$, we may obtain a formula for $\Psi^*(A)$ by extending linearly.
To this end, denote
\[
\hat{A}
:=A+\operatorname{tr}A\cdot I_d
=A-\frac{\operatorname{tr}A}{d}\cdot I_d
\in\{X\in\mathbb{F}_{q^2}^{d\times d}:\operatorname{tr}X=0\}.
\]
Then
\[
\Psi^*(A)
=\Psi^*(\hat{A}-\operatorname{tr}A\cdot I_d)
=\Psi^*(\hat{A})-\operatorname{tr}A\cdot \Psi^*(I_d)
=-b \hat{A}
=-b(A+\operatorname{tr}A\cdot I_d).
\]
In particular, we have
\[
\Psi(e_ie_j^*)
=\Big(\Psi^*(e_ie_j^*)\Big)^*
=\Big(-b(e_ie_j^*+\operatorname{tr}(e_ie_j^*)\cdot I_d)\Big)^*
=-b(e_je_i^*+\delta_{ij}\cdot I_d).
\]
Then Lemma~\ref{lem.ebr-esque move} gives
\begin{align*}
\sum_{k\in[n]}(x_k^{\otimes2})(x_k^{\otimes2})^*
&=\sum_{i\in[d]}\sum_{j\in[d]}e_ie_j^*\otimes \Psi(e_ie_j^*)\\
&=-2b\sum_{i\in[d]}\sum_{j\in[d]} e_ie_j^*\otimes \frac{1}{2}\Big( e_je_i^*+\delta_{ij}\cdot I_d \Big)
=-2b\cdot\Pi_d^{(2)}.
\end{align*}
Since $c_2:=-2b\neq-2a^2=0$ by assumption, it follows that $\{x_k\}_{k\in[n]}$ is an $(a,c_1,c_2)$-projective $2$-design, as desired.
\end{proof}

\subsection{A construction for Gerzon equality}

Theorem~\ref{thm.characterization of tight designs} allows one to easily identify projective $2$-designs over $\mathbb{F}_{q^2}^d$.
For example, \cite{GreavesIJM:20} constructs a $(0,1,0)$-ETF of $d^2$ vectors in $\mathbb{F}_{3^2}^d$ for every $d=2^{2\ell+1}$ with $\ell\in\mathbb{N}$.
Since $2^{2\ell+1}\equiv -1\bmod 3$, Theorem~\ref{thm.characterization of tight designs} implies that each of these systems of vectors forms a $(0,0,1)$-projective $2$-design for $\mathbb{F}_{3^2}^d$.
The following result constructs additional examples:

\begin{theorem}
\label{thm.gerzoneqfinitefield}
Select any prime $p$, positive integer $k$, and prime power $r$ such that $p$ divides $r-1$ and $r^2+r+1$ divides $p^k+1$.
Put $q:=p^k$ and $d:=r^2+r+1$.
Let $D\subseteq\mathbb{Z}/d\mathbb{Z}$ denote the Singer difference set, select a primitive element $\alpha\in\mathbb{F}_{q^2}^\times$, put $\omega:=\alpha^{(q^2-1)/d}$, and define translation and modulation operators by
\[
(Tf)(x):=f(x-1),
\qquad
(Mf)(x):=\omega^{x}\cdot f(x),
\qquad
f\colon \mathbb{Z}/d\mathbb{Z}\to\mathbb{F}_{q^2}.
\]
Then $\{M^sT^t\mathbf{1}_D\}_{s,t\in\mathbb{Z}/d\mathbb{Z}}$ is a $(2,1,2d)$-equiangular tight frame of $d^2$ vectors in $\mathbb{F}_{q^2}^d$.
\end{theorem}

The ETF construction in Theorem~\ref{thm.gerzoneqfinitefield} is a finite field analog of a biangular Gabor frame that was suggested in~\cite{BojarovskaP:15,HaasCTC:17}.
In the finite field setting, one might view this as a Steiner ETF~\cite{FickusMT:12} in which a harmonic ETF~\cite{Turyn:65,StrohmerH:03,XiaZG:05,DingF:07} plays the role of a ``flat'' simplex.
Empirically, there are many $(p,k,r)\in\mathbb{N}^3$ that satisfy the constraints that $p$ is prime, $r$ is a prime power, $p$ divides $r-1$, and $r^2+r+1$ divides $p^k+1$.
In fact, there are infinitely many, conditioned on the Lenstra--Pomerance--Wagstaff conjecture of the infinitude of Mersenne primes: whenever there exists a Mersenne prime $r=2^m-1$, we may take $p=2$ and $k=3m$.
We claim that the ETF construction in Theorem~\ref{thm.gerzoneqfinitefield} forms a projective $2$-design for $\mathbb{F}_{q^2}^d$ whenever $p>3$.
First, we have $a^2=4\neq 1=b$ and $a=2\neq0$, and so by Theorem~\ref{thm.characterization of tight designs}, it suffices to verify that $a^2-b=\frac{bc_1}{a}$.
Indeed, $a^2-b=4-1=3$ and $\frac{bc_1}{a}=d$, and since $p$ divides $r-1$ by assumption, we have $d=r^2+r+1\equiv 1+1+1=3\bmod p$, as desired.
The following table lists the smallest dimensions for which Theorem~\ref{thm.gerzoneqfinitefield} offers a construction, with gray columns indicating projective $2$-designs. 

\begin{center}
\begin{tabular}{|r| rr a rrrr a r a rr a r|}
\hline
$d$&13&57&73&307&757&993&1723&1723&2257&2257&2451&3541&3541&5113\\
$p$&2&2&7&2&2&2&2&5&2&23&2&2&29&2\\
$k$&6&9&12&51&378&15&287&287&90&30&63&118&590&213\\
$r$&3&7&8&17&27&31&41&41&47&47&49&59&59&71\\\hline
\end{tabular}
\end{center}

\begin{proof}[Proof of Theorem~\ref{thm.gerzoneqfinitefield}]
Proposition~6.1 in~\cite{GreavesIJM:20} gives that $\{M^sT^t\mathbf{1}_D\}_{s,t\in\mathbb{Z}/d\mathbb{Z}}$ is an $(a,da)$-NTF with
\[
a
=\langle \mathbf{1}_D,\mathbf{1}_D\rangle
=|D|
\equiv2\bmod p,
\]
where the last step applies the fact that $|D|=r+1\equiv 2\bmod p$.
It remains to verify equiangularity with parameter $b=1$.
To this end, we compute
\[
\langle M^sT^t\mathbf{1}_D,M^uT^v\mathbf{1}_D\rangle^{q+1}
\]
for every $(s,t)\neq(u,v)$ in two cases.

\textbf{Case I:} $t=v$.
Consider the discrete Fourier transform matrix $\mathcal{F}=[\omega^{ij}]_{i,j\in\mathbb{Z}/d\mathbb{Z}}\in\mathbb{F}_{q^2}^{d\times d}$, and observe that $\{M^sT^t\mathbf{1}_D\}_{s\in\mathbb{Z}/d\mathbb{Z}}$ equals the columns of $\operatorname{diag}(\mathbf{1}_{D+t})\mathcal{F}$, which in turn is a zero-padded version of the $|D|\times d$ submatrix $\mathcal{F}_{D+t}$.
Since $D+t$ is a difference set with parameter $\lambda=1$, (the proof of) Theorem~5.7 in~\cite{GreavesIJM:20} gives that $\mathcal{F}_{D+t}$ is a $(|D|,|D|-\lambda,d)$-ETF.
It follows that
\[
\langle M^sT^t\mathbf{1}_D,M^uT^t\mathbf{1}_D\rangle^{q+1}
=|D|-\lambda
\equiv 1\bmod p
\]
whenever $s\neq u$.

\textbf{Case II:} $t\neq v$.
We exploit the well-known fact that the so-called \textit{development}
\[
\{D+z:z\in G\}
\]
of a difference set $D\subseteq G$ with parameter $\lambda$ gives a symmetric block design in which every pair of blocks intersects in exactly $\lambda$ points (see Theorem~18.6 in~\cite{JungnickelPS:06}).
Since the Singer difference set has parameter $\lambda=1$, we have 
\begin{align*}
\langle M^sT^t\mathbf{1}_D,M^uT^v\mathbf{1}_D\rangle
&=\sum_{x\in\mathbb{Z}/d\mathbb{Z}}(M^sT^t\mathbf{1}_D)(x)^q(M^uT^v\mathbf{1}_D)(x)\\
&=\sum_{x\in\mathbb{Z}/d\mathbb{Z}}(\omega^{sx}\mathbf{1}_D(x-t))^q(\omega^{ux}\mathbf{1}_D(x-v))\\
&=\sum_{x\in\mathbb{Z}/d\mathbb{Z}}\omega^{(qs+u)x}\cdot\mathbf{1}_D(x-t)\mathbf{1}_D(x-v)
=\omega^{(qs+u)x_0},
\end{align*}
where $x_0$ is the unique member of the block intersection $(D+t)\cap(D+v)$.
Then
\[
\langle M^sT^t\mathbf{1}_D,M^uT^v\mathbf{1}_D\rangle^{q+1}
=(\omega^{(qs+u)x_0})^{q+1}
=1,
\]
as desired.
\end{proof}

\section{The quaternionic setting}

Consider the following generalization of Proposition~\ref{prop.proj2design}(c) for $\mathbb{F}\in\{\mathbb{R},\mathbb{C},\mathbb{H}\}$; see~\cite{Munemasa:07,Waldron:20}.
Put $m:=\frac{1}{2}\cdot[\mathbb{F}:\mathbb{R}]$ and $N=md$, and given $A,B\in\mathbb{F}^{d\times d}$, define $\langle A,B\rangle:=\operatorname{Re}\operatorname{tr}(A^*B)$.
We say unit vectors $\{x_k\}_{k\in[n]}$ in $\mathbb{F}^d$ form a \textbf{projective $2$-design} if
\[
\frac{1}{n^2}\sum_{k\in[n]}\sum_{\ell\in[n]}\langle x_kx_k^*,x_\ell x_\ell^*\rangle
=\frac{m}{N},
\qquad
\frac{1}{n^2}\sum_{k\in[n]}\sum_{\ell\in[n]}\langle x_kx_k^*,x_\ell x_\ell^*\rangle^2
=\frac{m(m+1)}{N(N+1)}.
\]
We say a projective $2$-design $\{x_k\}_{k\in[n]}$ for $\mathbb{F}^d$ is \textbf{tight} if it is also equiangular, which occurs precisely when equality is achieved in the lower bound $n\geq d+m(d^2-d)$~\cite{BannaiH:85}.
In this section, we use tight projective $2$-designs in quaternionic spaces to form nice arrangements of real subspaces.
Given $r$-dimensional subspaces $\{S_k\}_{k\in[n]}$ of a real vector space $V$, select an orthonormal basis $\{x_{ki}\}_{i\in[r]}$ for each $S_k$ and compute the cross-Gramians $G_{k\ell}:=[\langle x_{ki},x_{\ell j}\rangle]_{i,j\in[r]}$.
We say $\{S_k\}_{k\in[n]}$ is \textbf{equi-isoclinic} if there exists $\alpha\geq0$ such that $G_{k\ell}^*G_{k\ell}=\alpha I_r$ for every $k,\ell\in[n]$ with $k\neq\ell$.
We say $\{S_k\}_{k\in[n]}$ forms a \textbf{tight fusion frame} if
\[
\frac{1}{n^2}\sum_{k\in[n]}\sum_{\ell\in[n]}\|G_{k\ell}\|_F^2
=\frac{r^2}{\operatorname{dim}V}.
\]
An equi-isoclinic tight fusion frame is an optimal code in the Grassmannian $\operatorname{Gr}(r,V)$ under the \textit{chordal distance}, as it achieves equality in the \textit{simplex bound}~\cite{ConwayHS:96}.
It is also an optimal packing in terms of the \textit{spectral distance}, where we view each $S_k$ as a subset of projective space and seek to maximize the minimum distance between these subsets of this metric space~\cite{DhillonHST:08}.
One may rightly view such objects as analogs of equiangular tight frames.

Since quaternion multiplication is noncommutative (e.g., $\mathrm{i}\mathrm{j}=-\mathrm{j}\mathrm{i}$), we start by carefully  verifying a few simple facts that may be unfamiliar to the reader:

\begin{lemma}
$\mathbb{H}^{d\times d}$ is a real Hilbert space with inner product $\langle A,B\rangle=\operatorname{Re}\operatorname{tr}(A^*B)$.
\end{lemma}

\begin{proof}
First, given $u,v\in\mathbb{H}$, we observe that $\operatorname{Re}(\overline{u}v)$ equals the dot product between the corresponding real coordinate vectors:
\begin{align*}
\operatorname{Re}((a-b\mathrm{i}-c\mathrm{j}-d\mathrm{k})(e+f\mathrm{i}+g\mathrm{j}+h\mathrm{k}))
=ae+bf+cg+dh
=(a,b,c,d)\cdot(e,f,g,h).
\end{align*}
Next, given $A\in\mathbb{H}^{d\times d}$, let $\operatorname{vec}(A)\in\mathbb{R}^{4d^2}$ denote the vector of real coordinates of the entries of $A$.
Then the above observation gives
\[
\operatorname{Re}\operatorname{tr}(A^*B)
=\operatorname{Re}\sum_{j\in[d]}(A^*B)_{jj}
=\operatorname{Re}\sum_{j\in[d]}\sum_{i\in[d]} (A^*)_{ji}B_{ij}
=\sum_{i\in[d]}\sum_{j\in[d]} \operatorname{Re}\overline{A_{ij}}B_{ij}
=\operatorname{vec}(A)\cdot\operatorname{vec}(B).
\]
The result follows.
\end{proof}

\begin{lemma}
Suppose $A\in\mathbb{H}^{m\times n}$ and $B\in\mathbb{H}^{n\times m}$.
Then $\operatorname{Re}\operatorname{tr}(AB)=\operatorname{Re}\operatorname{tr}(BA)$.
\end{lemma}

\begin{proof}
Consider the algebra homomorphism $f\colon\mathbb{H}\to\mathbb{C}^{2\times2}$ defined by
\[
f(a+b\mathrm{i}+c\mathrm{j}+d\mathrm{k})
=\left[\begin{array}{rr}
a+b\mathrm{i}&c+d\mathrm{i}\\
-c+d\mathrm{i}&a-b\mathrm{i}
\end{array}\right].
\]
Observe that $\operatorname{tr}f(z)=2\operatorname{Re}z$.
Given $M\in\mathbb{H}^{n\times n}$, one may apply $f$ entrywise to obtain $f(M)\in\mathbb{C}^{2n\times 2n}$, and then $\operatorname{Re}\operatorname{tr}M = \frac{1}{2}\operatorname{tr}f(M)$.
Applying this to $AB$ and $BA$ gives
\[
\operatorname{Re}\operatorname{tr}(AB)
=\tfrac{1}{2}\operatorname{tr}(f(AB))
=\tfrac{1}{2}\operatorname{tr}(f(A)f(B))
=\tfrac{1}{2}\operatorname{tr}(f(B)f(A))
=\tfrac{1}{2}\operatorname{tr}(f(BA))
=\operatorname{Re}\operatorname{tr}(BA),
\]
as claimed.
\end{proof}

\begin{lemma}
Suppse $A,B\in\mathbb{H}^{d\times d}$ satisfy $A^*=A$ and $B^*=-B$.
Then $\langle A,B\rangle=0$.
\end{lemma}

\begin{proof}
Symmetry of the real inner product gives
\[
\langle A,B\rangle
=\langle B,A\rangle
=\operatorname{Re}\operatorname{tr}(B^*A)
=\operatorname{Re}\operatorname{tr}(-BA^*)
=-\operatorname{Re}\operatorname{tr}(A^*B)
=-\langle A,B\rangle.
\]
Rearrange to get the result.
\end{proof}

In fact, the real subspace of anti-Hermitian matrices is the orthogonal complement of the real subspace of Hermitian matrices.
This can be seen by dimension counting:
Hermitian matrices have dimension $d+4\binom{d}{2}$, while anti-Hermitian matrices have dimension $3d+4\binom{d}{2}$, and the sum of these is $4d^2$, i.e., the dimension of $\mathbb{H}^{d\times d}$.

Given $x\in\mathbb{H}^{d}\setminus\{0\}$, let $S(x)$ denote the $3$-dimensional real subspace of $\mathbb{H}^{d\times d}$ defined by
\[
S(x)
:=\{xzx^*:z\in\mathbb{H},~\operatorname{Re}z=0\}.
\]
For each $A\in S(x)$, it holds that
\[
A^*
=(xzx^*)^*
=x\overline{z}x^*
=-xzx^*
=-A.
\]
As such, $S(x)$ is actually a subspace of the anti-Hermitian matrices.
Observe that for each $u,v\in\{\mathrm{i},\mathrm{j},\mathrm{k}\}$, it holds that
\begin{align*}
\langle xux^*,xvx^*\rangle
&=\operatorname{Re}\operatorname{tr}((xux^*)^*xvx^*)\\
&=\operatorname{Re}\operatorname{tr}(x\overline{u}x^*xvx^*)\\
&=\operatorname{Re}\operatorname{tr}(\overline{u}x^*xvx^*x)
=\|x\|^4\cdot\operatorname{Re}(\overline{u}v)
=\left\{\begin{array}{cl}
\|x\|^4&\text{if }u=v\\
0&\text{otherwise.}
\end{array}\right.
\end{align*}
Thus, $\{x\mathrm{i}x^*,x\mathrm{j}x^*,x\mathrm{k}x^*\}$ is an orthogonal basis for $S(x)$.
With this, we prove the following:

\begin{theorem}
\label{thm.cousin object}
Consider unit vectors $\{x_k\}_{k\in[n]}$ in $\mathbb{H}^d$.
\begin{itemize}
\item[(a)]
If $\{x_k\}_{k\in[n]}$ is equiangular, then $\{S(x_k)\}_{k\in[n]}$ is equi-isoclinic.
\item[(b)]
If $\{x_k\}_{k\in[n]}$ is a projective $2$-design, then $\{S(x_k)\}_{k\in[n]}$ is tight in anti-Hermitian space.
\end{itemize}
\end{theorem}

\begin{proof}
Given unit vectors $x,y\in\mathbb{H}^d$, we are interested in the cross-Gramian $G_{xy}$ between $\{x\mathrm{i}x^*,x\mathrm{j}x^*,x\mathrm{k}x^*\}$ and $\{y\mathrm{i}y^*,y\mathrm{j}y^*,y\mathrm{k}y^*\}$.
The entry indexed by $(u,v)\in\{\mathrm{i},\mathrm{j},\mathrm{k}\}^2$ is given by
\begin{align}
\langle xux^*,yvy^*\rangle
=\operatorname{Re}\operatorname{tr}((xux^*)^*yvy^*)
\nonumber
&=\operatorname{Re}\operatorname{tr}(x\overline{u}x^*yvy^*)\\
\label{eq.quaternionic phase}
&=\operatorname{Re}\operatorname{tr}(\overline{u}x^*yvy^*x)
=\operatorname{Re}(\overline{u}x^*yvy^*x).
\end{align}
If $x^*y=0$, then $G_{xy}=0$.
Otherwise, define $z:=\frac{x^*y}{|x^*y|}$ and continue \eqref{eq.quaternionic phase}:
\[
\langle xux^*,yvy^*\rangle
=\operatorname{Re}(\overline{u}x^*yvy^*x)
=|x^*y|^2\cdot\operatorname{Re}(\overline{u}zvz^{-1}).
\]
It follows that $G_{xy}$ equals $|x^*y|^2$ times a matrix representation of the special orthogonal map $q\mapsto zqz^{-1}$ over imaginary $q\in\mathbb{H}$.
This implies (a).

For (b), we apply the facts that $G_{xy}=|x^*y|^2\cdot Q$ for some $Q\in\operatorname{SO}(3)$ and
\begin{align*}
|x^*y|^2
=\overline{x^*y}x^*y
=y^*xx^*y
=\operatorname{Re}\operatorname{tr}(y^*xx^*y)
=\operatorname{Re}\operatorname{tr}(xx^*yy^*)
=\langle xx^*,yy^*\rangle
\end{align*}
in order to compute the frame potential of $\{S(x_k)\}_{k\in [n]}$:
\begin{align*}
\frac{1}{n^2}\sum_{k\in[n]}\sum_{\ell\in[n]}\|G_{x_kx_\ell}\|_F^2
&=\frac{1}{n^2}\sum_{k\in[n]}\sum_{\ell\in[n]}3|x_k^*x_\ell|^4\\
&=\frac{3}{n^2}\sum_{k\in[n]}\sum_{\ell\in[n]}\langle x_kx_k^*,x_\ell x_\ell^*\rangle^2
=3\cdot\frac{m(m+1)}{N(N+1)}
=\frac{9}{d(2d+1)}.
\end{align*}
It follows that $\{S(x_k)\}_{k\in[n]}$ forms a tight fusion frame in the $d(2d+1)$-dimensional real vector space of $d\times d$ quaternionic anti-Hermitian matrices. 
\end{proof}

Theorem~\ref{thm.cousin object} implies that every tight projective $2$-design for $\mathbb{H}^d$ corresponds to an equi-isoclinic tight fusion frame of $d(2d-1)$ subspaces of $\mathbb{R}^{d(2d+1)}$ of dimension $3$.
To date, such designs are only known to exist for $d\in\{2,3\}$.
First, $\mathbb{HP}^1$ is isometric to $S^4$, and so the six vertices of a $5$-dimensional regular simplex easily deliver a tight projective $2$-design for $\mathbb{H}^2$; this in turn determines $6$ subspaces of $\mathbb{R}^{10}$ of dimension $3$.
The $d=3$ case is resolved by Theorem~4.12 in~\cite{CohnKM:16}, which uses a variant of the Newton--Kantorovich theorem to obtain a computer-assisted proof of the existence of a $15$-point simplex in $\mathbb{HP}^2$; this determines $15$ subspaces of $\mathbb{R}^{21}$ of dimension $3$.
The authors are not aware of any other constructions of equi-isoclinic tight fusion frames with these parameters.
Also, it is open whether tight projective $2$-designs exist for $\mathbb{H}^d$ with $d>3$.

\section*{Acknowledgments}

The first part of this paper was inspired by a beautiful talk~\cite{Paulsen:online} given by Vern Paulsen at the Codes and Expansions online seminar.
Much of this work was conducted during the SOFT 2020:\ Summer of Frame Theory virtual workshop. DGM was partially supported by AFOSR FA9550-18-1-0107 and NSF DMS 1829955.

\end{document}